\documentclass[a4paper,10pt,twoside]{article} 
\usepackage{ifpdf}
\ifpdf\else
  \PassOptionsToPackage{dvipdfmx}{graphicx}
  \PassOptionsToPackage{dvipdfmx}{color}
\fi
\usepackage{amsmath, amssymb, amsthm}
\usepackage{color}

\newif\ifdraftfigures
\draftfiguresfalse 
\ifdraftfigures
  \usepackage[draft]{graphicx}
\else
  \usepackage{graphicx}
\fi
\usepackage{bm}

\theoremstyle{definition}

\newtheorem{thm}{Theorem}[section]
\newtheorem{prop}{Proposition}[section]
\newtheorem{lemma}{Lemma}[section]
\newtheorem{cor}{Corollary}[section]
\newtheorem{remark}{Remark}[section]

\newtheorem{defn}{Definition}[section]

\numberwithin{equation}{section}
\allowdisplaybreaks

\def\R{\mathbb{R}}

\def\DD{\mathbb{D}}

\def\N{\mathbb{N}}

\def\P{{\mathbb{P}}}

\def\toP{\stackrel{p}{\longrightarrow}}

\def\E{\mathbb{E}}

\def\wh#1{\widehat{#1}}

\def\Var{\mathrm{Var}}

\title{M-Estimation based on quasi-processes from discrete samples of L\'evy processes}
\author{Yasutaka Shimizu\footnote{Department of Applied Mathematics, Waseda University. E-mail:{\tt shimizu@waseda.jp}}\ \ 
and Hiroshi Shiraishi\footnote{Department of Mathematics, Keio University. E-mail:{\tt shiraishi@math.keio.ac.jp}}}
\date{\today}

\begin{document}
\maketitle
\begin{abstract}
We propose a novel estimation framework for path-dependent functionals of L\'evy processes from discretely observed data. Traditional approaches rely on Monte Carlo simulation of full paths, which requires complete model specification and heavy computation. In contrast, our \emph{quasi-process} method constructs pseudo-paths directly from observed increments by random permutation, preserving the increment distribution while enabling repeated evaluation of functionals. {Under a high-frequency, long-term sampling regime, we establish unconditional weak convergence of the quasi-process and consistency of the uniformly randomized M-estimator under finite-horizon approximation and identification conditions. A separate normality theorem requires a central limit theorem for the conditional permutation score; this additional condition is verified directly for a quadratic endpoint contrast, but remains to be established for the default and dividend examples.} Numerical experiments with the original, unsmoothed dividend criterion examine finite-sample concentration and dispersion beyond the verified scope of these sufficient conditions; they do not assert consistency or asymptotic normality for that criterion.
\begin{flushleft}
{ \it Keywords:} M-estimation; L\'evy processes; resampling; quasi-process; discrete observations; weak convergence. 
\vspace{1mm}\\
{\it MSC2010:} {\bf 62M20}; 60G51, 62G20. 
\end{flushleft}
\end{abstract}

\section{Introduction}\label{sec:intro}
L\'evy processes play a central role in stochastic modeling, with applications in finance, insurance, and risk theory. Many problems in these fields involve the estimation of \emph{path-dependent functionals} of a process from observed data. Such functionals arise naturally in the evaluation of financial derivatives, in the optimization of dividend or reinsurance strategies, and in the assessment of ruin-related quantities. A key feature of these problems is that the target parameter is defined through an expectation over the \emph{entire path} of the process, making direct estimation from limited discrete observations challenging.

Let $(\Omega, \mathcal{F}, \mathbb{P}; \mathbb{F})$ be a stochastic basis with filtration $\mathbb{F}=(\mathcal{F}_t)_{t \ge 0}$, and let $X=(X_t)_{t \ge 0}$ be an $\mathbb{F}$-L\'evy process starting at $X_0=u$ of the form
\[
X_t = u + \widetilde{X}_t,
\]
where $\widetilde{X}$ has the characteristic exponent
\begin{align}
\log \mathbb{E}[e^{it \widetilde{X}_1}] 
=  i \mu t - \frac{\sigma^2}{2}t^2 
+ \int_{\mathbb{R}} \left( e^{i tz} - 1 - itz \, \mathbf{1}_{\{|z| \le 1\}} \right) \nu(\mathrm{d} z). \label{ch}
\end{align}
Here $(\mu,\sigma^2,\nu)$ is the L\'evy triplet relative to the
truncation function $z\mathbf1_{\{|z|\le1\}}$: $\mu\in\R$ is the drift
coefficient for this truncation, $\sigma^2\ge0$ is the variance per unit
time of the Gaussian component, and $\nu$ is the L\'evy measure, a
nonnegative Borel measure on $\R$ satisfying
\[
 \nu(\{0\})=0,\qquad \int_{\R}(1\wedge z^2)\,\nu(dz)<\infty.
\]
For a Borel set $A$ bounded away from zero, $\nu(A)$ is the expected
number per unit time of jumps with sizes in $A$. The drift coefficient
$\mu$ need not equal the mean increment: when
$\int_{|z|>1}|z|\,\nu(dz)<\infty$,
\[
 \E[X_t-X_0]=t\left(\mu+\int_{|z|>1}z\,\nu(dz)\right).
\]
Let $\mathbb{D}_T := D([0,T])$ be the space of c\`adl\`ag functions on {$[0,T]$}, equipped with the Borel $\sigma$-field $\mathcal{D}_T$ generated by the Skorokhod topology; see Billingsley~\cite{b99}. For a measurable functional $f: \mathbb{D}_T \to \mathbb{R}$, we write
\[
Pf := \int_{\mathbb{D}_T} f(x) \, P(\mathrm{d}x) = \mathbb{E}[f(X)], \quad P := \mathbb{P} \circ X^{-1}.
\]

Let $\Theta \subset \mathbb{R}^d$ and consider a measurable map {$h: \mathbb{D}_T \times \Theta \to \mathbb{R}$, with $h_\theta(x)=h(x,\theta)$}. The \emph{true parameter} is defined as
\begin{align}
\theta_0 = \arg\min_{\theta \in \Theta} P h_\theta. \label{eq:true}
\end{align}
Our goal is to estimate $\theta_0$ from discrete observations of $X$. {For an expected reward formulated as a maximization below, the corresponding minimum contrast is its negative.}

Such a formulation encompasses many problems of practical relevance. For instance, in finance, consider the optimal exercise strategy for a \emph{perpetual American put option} with strike price $K>0$. If the holder exercises when $X_t < \theta$ for some threshold $\theta>0$, with stopping time $\tau^\theta := \inf\{ t>0 : X_t < \theta \}$, {the best threshold within this class of exercise rules is defined by}
\begin{align}
\theta_* = \arg\max_{\theta \in [0,K]} 
\mathbb{E}\left[ e^{-r \tau^\theta} (K - X_{\tau^\theta})_+ \mathbf{1}_{\{\tau^\theta < \infty\}} \right], \label{call}
\end{align}
where $r>0$ is the interest rate; see Gerber and Shiu~\cite{gs97}. In insurance mathematics, one may consider the expected present value of \emph{ruin-related losses} up to the ruin time $\tau^\theta := \inf\{t>0: X_t < f(\theta)\}$:
\[
h_\theta(X) = \int_0^{\tau^\theta} e^{-rt} U_\theta(X_t) \, \mathrm{d}t,
\]
where $f:\Theta \to \mathbb{R}$ and $U_\theta$ specifies the loss or dividend structure; see Feng~\cite{f11} and Feng and Shimizu~\cite{fs13}. For example, {for a dividend-reward functional with payment rate $a>0$} whenever $X_t \ge \theta$, we have
\begin{align}
Ph_\theta = a \int_0^\infty e^{-rt} \mathbb{P}\left(X_t \ge \theta, \, t \le \tau^\theta\right) \, \mathrm{d}t. \label{dividend}
\end{align}

In these examples, the objective functionals \eqref{call} and \eqref{dividend} involve the stopping time $\tau^\theta$, which depends on the entire path of $X$ and is generally unobservable from discrete data. If one could simulate many independent full paths of $X$, the functionals could be approximated via Monte Carlo methods. This classical simulation-based approach has been widely used; however, it requires complete knowledge of the model and extensive computational resources. In contrast, practical applications usually involve only a single discretely observed trajectory, and model uncertainty often makes direct simulation infeasible. This mismatch between theoretical formulations and practical data availability is a central obstacle in modern stochastic modeling.

\medskip
\noindent
\textbf{Our contribution.} In this paper we propose a \emph{quasi-process} method that constructs pseudo-sample paths from discrete observations. Given $X_{t_0}, X_{t_1}, \dots, X_{t_n}$ with $t_k = kh_n$ and increments $\Delta_k X = X_{t_k} - X_{t_{k-1}}$, we form step functions jumping at $t_k$ with size $\Delta_k X$. By randomly permuting these increments, we generate artificial paths that preserve the increment distribution of $X$ but allow repeated evaluation of path functionals. This bootstrap-like procedure enables estimation of $P h_\theta$ and construction of an M-estimator for $\theta_0$ without simulating the true process. Compared with traditional Monte Carlo simulation, our method directly reuses the observed data, avoids assumptions on the full model dynamics, and provides a novel route to inference when only one observed trajectory is available.

\medskip
\noindent
\textbf{Theoretical properties.} Under a \emph{high-frequency, long-term} (HFLT) sampling regime, where $h_n \to 0$ and $T_n = n h_n \to \infty$ as $n \to \infty$, we establish weak convergence of the quasi-process to the true L\'evy process in the Skorokhod topology. {For independently sampled uniform permutations, we prove a uniform law of large numbers using a variance bound for finite-horizon permutation averages. This yields consistency under explicit approximation and identification conditions. We give an asymptotic normality theorem conditional on a data-score central limit theorem and local Hessian control. The former pathwise approximation argument does not establish those hypotheses for general stopped functionals. For such applications, consistency and normality must therefore be distinguished.
}

\medskip
\noindent
\textbf{Practical implications and future perspectives.} Beyond its theoretical elegance, the proposed methodology has broad potential in applications. In finance, it may serve as a robust tool for pricing and hedging exotic derivatives, where payoffs depend on complex stopping times or cumulative dividends. In insurance, it provides a new way to evaluate dividend strategies, reinsurance treaties, and ruin-related risk measures from high-frequency data. Since the quasi-process approach requires no simulation of the underlying dynamics, it can substantially reduce computational costs and enhance stability in large-scale actuarial evaluations. Moreover, because the method relies only on increment distributions, it is flexible enough to be adapted to processes with jumps, stochastic volatility, or even heavy-tailed behavior.

Looking forward, the quasi-process methodology opens the door to several promising research directions. One avenue is the extension to multivariate L\'evy processes, enabling joint modeling of multiple risk sources such as equity, credit, and insurance portfolios. Another is its integration into machine learning frameworks, where quasi-process resampling could serve as a data augmentation technique for path-dependent features. Finally, in regulatory and risk management contexts, the approach offers a principled way to link discrete financial or insurance records to theoretically grounded path-based risk measures, potentially informing solvency assessments and capital allocation under modern regimes such as Solvency II.

The remainder of the paper is organized as follows. Section~\ref{sec:quasi-proc} defines the quasi-process and the associated M-estimator. Section~\ref{sec:weak} establishes the weak convergence of quasi-processes under HFLT sampling and illustrates this via simulations. Section~\ref{sec:main} presents the main statistical results, including consistency and asymptotic normality, along with sufficient conditions. {In Section~\ref{sec:example}, we study continuity of stopped loss functionals, identify invalid pathwise rate claims, give a directly verifiable normal limit, and state the remaining conditions for dividend applications.} Section~\ref{sec:dividend-simulation} studies the original, unsmoothed dividend M-estimator numerically, examining concentration and dispersion relevant to consistency and limitations of a normal approximation beyond the verified scope of the sufficient conditions. Finally, we conclude with remarks on further applications and possible extensions in Section \ref{sec:conclude}.

\subsection*{Notation}

Throughout the paper, we use the following notation:

{We write $\leadsto$ for convergence in distribution, $\toP$ for convergence in probability, and $\stackrel{d}{=}$ for equality in law. Path-space convergence uses the Skorokhod $J_1$ topology. Unless explicitly conditioned on the observations, convergence of randomized estimators and $L^p$ norms of randomized quantities refer to the joint law of the process and auxiliary permutations; $P$ continues to denote the path law of $X$.}

\begin{itemize}
\item For two nonnegative quantities $A$ and $B$, we write $A \lesssim B$ if there exists a universal constant $c>0$ such that $A \le c  B$. 

\item For a $k$-th order tensor $x=(x_{i_1,i_2,\dots,i_k})_{i_1,\dots,i_k=1,\dots,d} \in \mathbb{R}^d \otimes \cdots \otimes \mathbb{R}^d$, we define the Euclidean norm
\[
|x| = \left( \sum_{i_1=1}^d \cdots \sum_{i_k=1}^d x_{i_1,i_2,\dots,i_k}^2 \right)^{1/2}.
\]

\item For $T>0$ and $x \in \mathbb{D}_T \cup \mathbb{D}_\infty$, we write
\[
\|x\|_T = \sup_{t \in [0,T]} |x(t)|, 
\quad 
\|x\|_\infty = \sup_{t \in [0,\infty)} |x(t)|.
\]

\item For $x,y \in \mathbb{D}_T$, we denote by $\varrho_T(x,y)$ the \emph{Skorokhod metric}:
\[
\varrho_T(x,y) = \inf_{\lambda \in \Lambda_T} \max \left\{ \| x \circ \lambda - y \|_T, \ \|\lambda - I\|_T \right\},
\]
{Here $\Lambda_T$ is the set of strictly increasing continuous bijections of $[0,T]$ onto itself, and $I(t)=t$.}

\item For $p \ge 1$ and a measure $Q$ on a measurable space $\mathcal{X}$, we define
\[
\|f\|_{L^p(Q)} := \Bigg( \int_{\mathcal{X}} |f(x)|^p \, Q(\mathrm{d}x) \Bigg)^{1/p}, \quad f:\mathcal{X} \to \mathbb{R}.
\]
The space $L^p(Q)$ consists of all measurable functions $f$ with $\|f\|_{L^p(Q)}<\infty$. For $f \in L^1(Q)$, we write
\[
Qf := \int_{\mathcal{X}} f(x) \, Q(\mathrm{d}x).
\]

\item For a function $f:\mathcal{X}\times \mathcal{Y} \to \mathbb{R}$, we write $f \in C^{n,m}(\mathcal{X}\times \mathcal{Y})$ if $f$ is $n$-times continuously differentiable with respect to $x\in \mathcal{X}$ and $m$-times continuously differentiable with respect to $y\in \mathcal{Y}$. In particular, $C^0$ denotes the set of continuous functions. Moreover, $f \in C_b^{n,m}$ means that $f \in C^{n,m}$ and that $f$ together with all partial derivatives up to order $(n,m)$ are bounded.
\end{itemize}

\section{Contrast functions via ``quasi-processes"}\label{sec:quasi-proc}

We observe a L\'evy process $X$ starting from $u$ at equidistant time points: 
\[
\{X_{t_k}\}_{k=0,1,\dots,n}, \quad 0 = t_0 < t_1 < \cdots < t_n = {T_n}, \quad h_n := t_k - t_{k-1}. 
\] 
Let $\mathbb{X} = (\Delta_1 X, \Delta_2 X, \dots, \Delta_n X)$ denote the vector of increments, where $\Delta_k X := X_{t_k} - X_{t_{k-1}}$.  
Consider the set of all permutations of $(1,2,\dots,n)$: 
\[
\Lambda_n := \Bigg\{ i_m = 
\begin{pmatrix}
1 & 2 & \dots & n \\
i_m(1) & i_m(2) & \dots & i_m(n)
\end{pmatrix} 
\,\Bigg|\, m=1,2,\dots,n! \Bigg\}.
\]
Since the increments $(\Delta_k X)_{1 \le k \le n}$ are i.i.d., the vector $\mathbb{X}$ is \emph{exchangeable}; that is, for any permutation $i \in \Lambda_n$,  
\[
i(\mathbb{X}) := (\Delta_{i(1)}X, \dots, \Delta_{i(n)}X)
\]
has the same distribution as $\mathbb{X}$.

\begin{defn}\label{def:quasi-path}
For given increments $\mathbb{X}$ and a permutation $i \in \Lambda_n$,  
the process $\wh{X}^n(i) = (\wh{X}_t^n(i))_{t \ge 0}$ defined by
\[
\wh{X}_t^n(i) = u + \sum_{k=1}^n \Delta_{i(k)}X \cdot \mathbf{1}_{[t_k,\infty)}(t)
\]
is called the \emph{quasi-process of $X$ associated with $i \in \Lambda_n$}.
\end{defn}

The path of a quasi-process $\wh{X}^n(i)$ is a right-continuous step function in $\mathbb{D}_\infty$, with jumps of size $\Delta_{i(k)}X$ occurring at times $t_k$ ($k=1,2,\dots,n$).  
{For each $n$, let $I_{n,1},\ldots,I_{n,a_n}$ be i.i.d. uniform random
elements of $\Lambda_n$, independent of the entire process $X$.
Sampling is with replacement: repeated permutations are counted separately.
The auxiliary draws can be constructed on a product extension of the
original probability space. Define
\[
 \mathcal G_n=\sigma(\Delta_1X,\ldots,\Delta_nX),\qquad
 A_n=(I_{n,1},\ldots,I_{n,a_n}),
\]
where $A_n$ is an indexed list, not a subset with duplicates removed.
The randomized empirical measure is
\begin{equation}\label{Pn}
 P_n=P_{A_n}:=\frac1{a_n}\sum_{b=1}^{a_n}\delta_{\wh{X}^n(I_{n,b})}.
\end{equation}
For measurable integrable $f$, its conditional permutation average is
\[
 \widehat P_nf:=\E[f(\wh{X}^n(I_{n,1}))\mid\mathcal G_n]
       =\frac1{n!}\sum_{i\in\Lambda_n}f(\wh{X}^n(i)).
\]
This is a random variable determined by the data. The draws are i.i.d.
conditionally on $\mathcal G_n$; unconditionally they share the same
increments and are generally dependent. Uniformity is needed for the
statistical results in Section~\ref{sec:main}. In contrast, the marginal
weak convergence in Section~\ref{sec:weak} holds for any deterministic
permutation, or any random permutation independent of $X$.

\begin{defn}
For the randomized measure \eqref{Pn}, a measurable minimum contrast
estimator is
\begin{equation}\label{M-est}
 \widehat\theta_n\in\arg\min_{\theta\in\Theta}P_nh_\theta.
\end{equation}
Existence and a measurable choice are assumed. The consistency theorem
uses a compact parameter set and continuous sample contrasts.
\end{defn}
}

In certain cases, $P h_\theta$ reduces to a functional of the L\'evy measure $\nu$ of the form
\[
P h_\theta = \int_{\mathbb{R}} H_\theta(z)\,\nu(\mathrm{d}z),
\]
for some function $H_\theta:\mathbb{R}\to\mathbb{R}$. In such situations, $P h_\theta$ can be estimated directly as
\[
\widehat{P h_\theta} = \frac{1}{n h_n}\sum_{k=1}^n H_\theta^n(\Delta_k X)   \toP  P h_\theta,
\]
with $H^n_\theta \to H_\theta$ under appropriate regularity conditions; see Comte and Genon-Catalot~\cite{cg11}, Jacod~\cite{j07}, Shimizu~\cite{s09}, and Kato and Kurisu~\cite{kk20}.  
In such settings, the quasi-process construction is unnecessary because the estimator is already exchangeable with respect to increments.  
Our approach is particularly useful when $h_\theta$ depends on the \emph{entire path} of $X$. For example, in the dividend problem \eqref{dividend}, many observations satisfying $X_t \ge \theta$ may be required. Using quasi-processes ${\{\wh{X}^n(I_{n,b})\}_{b=1}^{a_n}}$, one can generate numerous ``quasi-samples'' with $\wh{X}_t^n(i) \ge \theta$, {thereby reducing Monte Carlo error in the conditional permutation average. Its approximation of the population expectation requires a separate statistical argument.}

\section{Weak convergence for quasi-processes}\label{sec:weak}
\subsection{Theoretical results}
We consider the following \emph{high-frequency long-term} (HFLT) sampling scheme:
\begin{flushleft}
\textbf{(HFLT)} \quad $h_n \to 0$ and $T_n := n h_n \to \infty$ as $n \to \infty$.
\end{flushleft}

{\begin{thm}\label{thm:wc}
For each $n$, let $i_n\in\Lambda_n$ be deterministic, or random independently
of $X$ (its law need not be uniform). Under (HFLT),
\[
 \wh{X}^n(i_n)\leadsto X\quad\hbox{in }\DD_\infty.
\]
This is an unconditional statement, not convergence of the conditional
law given the observations.
\end{thm}
\begin{proof}
We first identify the law of the quasi-process, and then establish a
pathwise approximation for the discretized original process.
For $v=(v_1,\ldots,v_n)\in\R^n$, define the step-path map
\[
 \Phi_n(v)(t)=u+\sum_{k=1}^n v_k\mathbf1_{[kh_n,\infty)}(t),
 \qquad t\ge0.
\]
For any deterministic $i\in\Lambda_n$, exchangeability yields
\[
 (\Delta_{i(1)}X,\ldots,\Delta_{i(n)}X)
 \stackrel{d}{=}(\Delta_1X,\ldots,\Delta_nX).
\]
Applying $\Phi_n$ to both vectors gives equality of laws on $\DD_\infty$:
\[
 \wh{X}^n(i)\stackrel{d}{=}X^{[n]},\qquad
 X^{[n]}_t:=X_{(\lfloor t/h_n\rfloor\wedge n)h_n}.
\]
If $i_n$ is random and independent of $X$, then for every bounded
measurable functional $F$ on $\DD_\infty$,
\[
 \E F(\wh{X}^n(i_n))
 =\sum_{i\in\Lambda_n}\P(i_n=i)\E F(\wh{X}^n(i))
 =\E F(X^{[n]}).
\]
Thus it suffices to prove $X^{[n]}\to X$ almost surely in local $J_1$.

Fix $T>0$ and a c\`adl\`ag path $y$ on $[0,T]$ continuous at $T$.
Write $S_ny(t)=y(\lfloor t/h_n\rfloor h_n)$ for $0\le t\le T$.
For every $\varepsilon>0$, the regulated-function property supplies a
right-continuous step function $z$ with finitely many jump times
$0<s_1<\cdots<s_m<T$, no jump at $T$, and
\[
 \|y-z\|_T<\varepsilon,\qquad
 \|S_ny-S_nz\|_T\le\|y-z\|_T<\varepsilon.
\]
For all sufficiently large $n$, the rounded jump times
\[
 r_{j,n}:=h_n\lceil s_j/h_n\rceil\quad(1\le j\le m)
\]
are distinct and lie in $(0,T)$. They satisfy
$0\le r_{j,n}-s_j<h_n$. Define $\lambda_n$ by linear interpolation
through $(0,0)$, $(s_j,r_{j,n})$ for $1\le j\le m$, and $(T,T)$.
Then $\lambda_n\in\Lambda_T$ and
\[
 \|\lambda_n-I\|_T\le h_n,\qquad
 (S_nz)(\lambda_n(t))=z(t)\quad(0\le t\le T).
\]
Indeed, each interval between consecutive jumps of $z$ is mapped to
an interval on which $S_nz$ has the same constant value. If $z$ has
no jumps, take $\lambda_n=I$. It follows that
\begin{align*}
 \|S_ny\circ\lambda_n-y\|_T
 &\le\|(S_ny-S_nz)\circ\lambda_n\|_T
     +\|S_nz\circ\lambda_n-z\|_T+\|z-y\|_T\\
 &<2\varepsilon,
\end{align*}
so that
\[
 \varrho_T(S_ny,y)\le\max\{h_n,2\varepsilon\},\qquad
 \limsup_{n\to\infty}\varrho_T(S_ny,y)\le2\varepsilon.
\]
Since $\varepsilon$ is arbitrary, $\varrho_T(S_ny,y)\to0$.
This is $J_1$ convergence, allowing the time changes needed at jumps.

A L\'evy process has c\`adl\`ag paths and no jump at a fixed time almost
surely. Consequently,
\[
 \P\bigl(X\text{ is continuous at every }T\in\N\bigr)=1.
\]
On this event, for each integer $T\ge1$ we have $T_n\ge T$ eventually,
$X^{[n]}|_{[0,T]}=S_n(X|_{[0,T]})$, and hence
\[
 \varrho_T(X^{[n]}|_{[0,T]},X|_{[0,T]})\longrightarrow0.
\]
Convergence on these unbounded continuity horizons implies convergence
in the local $J_1$ topology on $\DD_\infty$; see
Billingsley~\cite{b99}, Sections 12 and 16. Therefore
\[
 X^{[n]}\longrightarrow X\quad\text{almost surely in local }J_1,
 \qquad \wh{X}^n(i_n)\stackrel{d}{=}X^{[n]}\leadsto X,
\]
which proves the assertion.
\end{proof}
}
\subsection{Numerical experiments for generating quasi-paths}

We distinguish unconditional path-law convergence from approximation of the
conditional permutation law based on one observed trajectory. Consider
\begin{align}
 X_t=u+c t+\sigma W_t-\sum_{j=1}^{N_t}\xi_j,\label{jd}
\end{align}
where $W$ is Brownian motion, $N$ is a Poisson process of rate $\lambda$,
and the jump sizes are independent exponential random variables with mean
$m$, all mutually independent. We use
$(c,\sigma,\lambda,m,u)=(20,10,5,3,0)$ throughout this subsection.
Here $c$ is the deterministic drift in the uncompensated compound-Poisson
representation \eqref{jd}. It differs from the triplet coefficient in
\eqref{ch}: for this model,
\[
 \nu(dz)=\frac{\lambda}{m}e^{z/m}\mathbf1_{\{z<0\}}\,dz,
 \qquad \mu=c+\int_{|z|\le1}z\,\nu(dz).
\]
In particular, $\E X_1=c-\lambda m=5$ and
$\Var(X_1)=\sigma^2+2\lambda m^2=190$.

The experiment uses exact increment sampling on a grid of width
$h_0=0.001$: each increment is
\[
 c h_0+\sigma\sqrt{h_0}Z-G_K,\qquad
 Z\sim N(0,1),\quad K\sim\operatorname{Poisson}(\lambda h_0),
\]
where $G_K$ has a gamma distribution with shape $K$ and scale $m$
conditionally on $K$, with $G_0=0$. The normal term is independent of the
jump term. A single baseline on $[0,100]$ is generated first. Every
observation window $T_n\in\{10,50,100\}$ is a prefix of this same path,
and coarser increments are sums of consecutive fine-grid increments.
Thus all panels use the same underlying realization, rather than new
data for each setting. The plotted paths are right-continuous step
interpolations of exact grid observations. Separate random streams are used for the baseline, independent
paths, permutations, and marginal reference sample.

\begin{figure}[tbp]
\centering
\includegraphics[width=\linewidth]{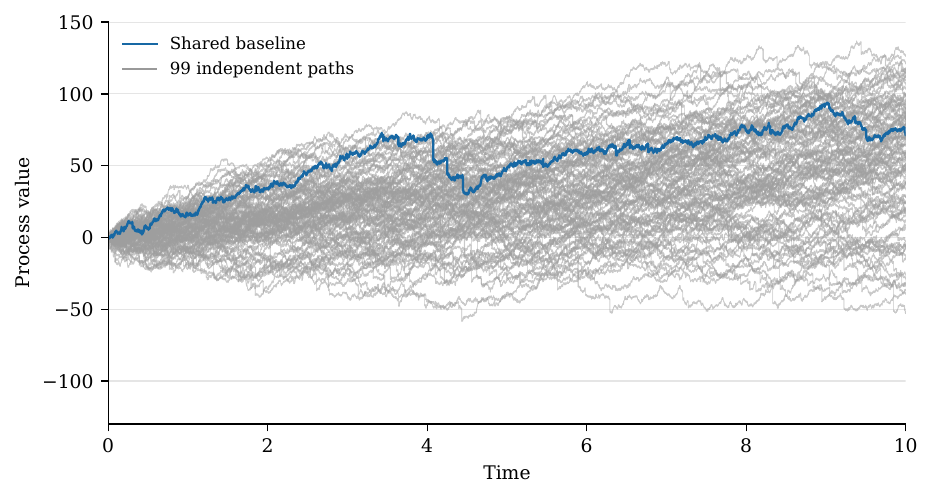}
\caption{Exact grid observations from \eqref{jd}, displayed on $[0,10]$
with grid width $0.001$. Blue: the common baseline used in every subsequent
path panel. Gray: 99 independent realizations. These are simulated data.}
\label{fig:path}
\end{figure}

Figure~\ref{fig:quasi-grid} compares 100 independent uniform permutation
draws in each setting. Indices are sampled without replacement within
each permutation; whole permutations are drawn independently with
replacement, as in Section~\ref{sec:quasi-proc}. The first 10 time units
are displayed on identical axes. The blue baseline is shown at its fine
grid resolution even when the gray paths use coarse observations.
When $T_n=10$, every complete permutation must end at the observed
$X_{10}$, since the sum of all increments is unchanged. The concentration
of gray endpoints in the left panels is therefore a conditional constraint,
not evidence that these paths are independent realizations of $X$.

\begin{figure}[tbp]
\centering
\includegraphics[width=\linewidth]{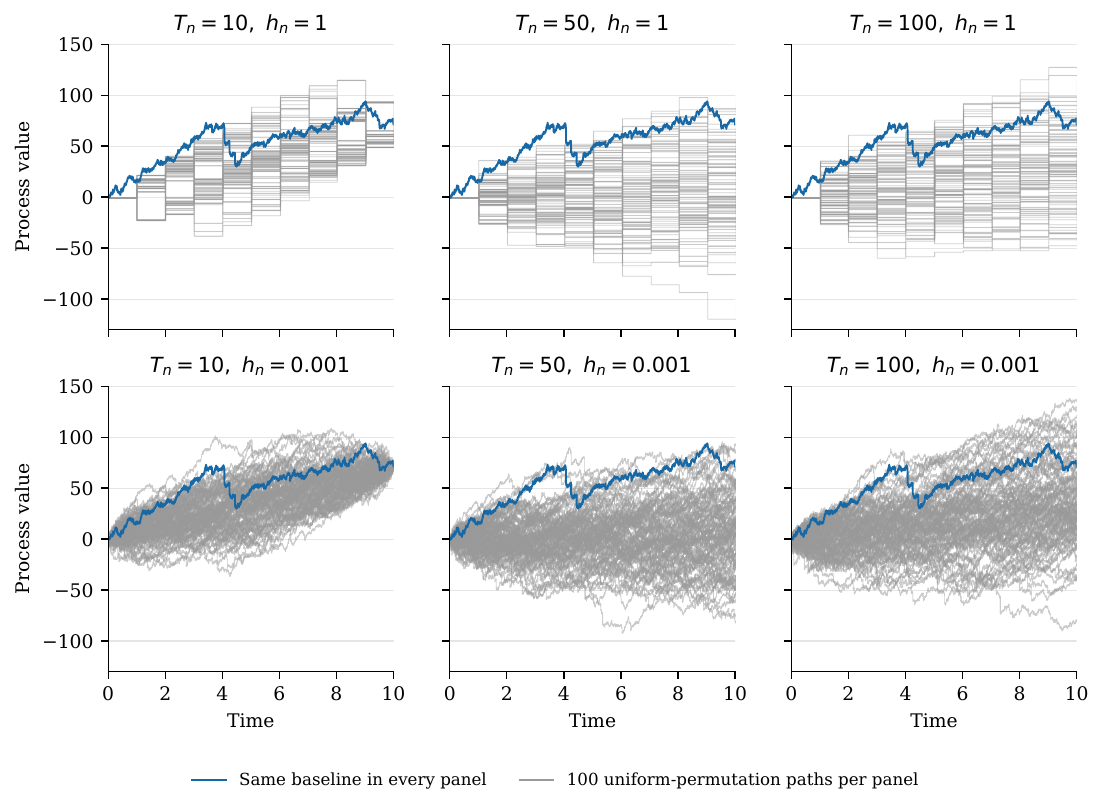}
\caption{Conditional quasi-paths from one shared baseline. Columns:
$T_n=10,50,100$. Rows: $h_n=1$ and $h_n=0.001$. Each panel contains
100 uniform permutation draws in gray and the same fine-grid baseline
in blue. The orange dashed line joins the exact conditional means at
observation grid times, $u+(t_k/T_n)(X_{T_n}-u)$; between grid times the
quasi-process mean is constant. All panels display $[0,10]$ with the
same vertical scale.}
\label{fig:quasi-grid}
\end{figure}

\begin{remark}[Center of the conditional quasi-path distribution]
\label{rem:quasi-center}
Permuting increments does not add a centered perturbation to the observed
path. Given $\mathcal G_n$, each increment has the same chance of appearing
in any position. Thus, with $\bar d=n^{-1}\sum_{j=1}^n\Delta_jX$,
\[
 \E[\wh{X}_{t_k}^n(I)\mid\mathcal G_n]
 =u+k\bar d=u+\frac{t_k}{T_n}(X_{T_n}-u).
\]
The conditional means at grid times lie on the line joining the observed
start and endpoint over the full observation horizon $T_n$. They need
not coincide with the observed intermediate values $X_{t_k}$, and the
conditional distribution need not be symmetric about its mean.
For the realization in Figure~\ref{fig:quasi-grid}, $X_5=48.31$, whereas
the conditional means at time 5 are $35.68$, $4.52$, and $16.36$ for
$T_n=10,50,100$, respectively. The same blue baseline is used in all
panels; its early segment lies relatively high compared with these
conditional mean lines. This explains why the gray paths tend to appear
below it in each setting, rather than indicating a systematic downward
bias of the permutation construction. Unconditionally,
$\wh{X}_{t_k}^n(I)\stackrel{d}{=}X_{t_k}$ by exchangeability. For
$T_n=10$, all permuted paths still share the blue endpoint at time 10.
\end{remark}

We next compare the distribution at the exact time $t=1$.
For each of the nine settings $T_n\in\{10,50,100\}$ and
$h_n\in\{1,0.01,0.005\}$, Figure~\ref{fig:marginal-cdf} shows an
empirical CDF from 20,000 independent conditional permutation draws.
Computationally, a uniform subset of $k=1/h_n$ distinct increments is
summed for each draw. This has exactly the same distribution as the
first $k$ increments of a uniform permutation and avoids generating
its unused tail. The black reference CDF is computed from 200,000
independent exact draws of $X_1$, independently of the baseline.
It is a Monte Carlo reference, not an analytic CDF. Its light gray
band has half-width $\sqrt{\log(200)/(2\cdot200000)}<0.00364$,
the 99\% simultaneous Dvoretzky--Kiefer--Wolfowitz bound with the sharp constant of Massart~\cite{massart90}. The top row shows the full CDFs. To make the narrow band visible, the middle row subtracts the reference CDF and enlarges the vertical scale for the two finer meshes. The gray region represents uncertainty of the independent reference only, not a confidence band for the conditional CDFs or their differences. The bottom row directly displays the difference between the two finer-mesh CDFs; no confidence band is asserted for this difference. No kernel
smoothing is applied to either reference or conditional samples.

\begin{figure}[tbp]
\centering
\includegraphics[width=\linewidth]{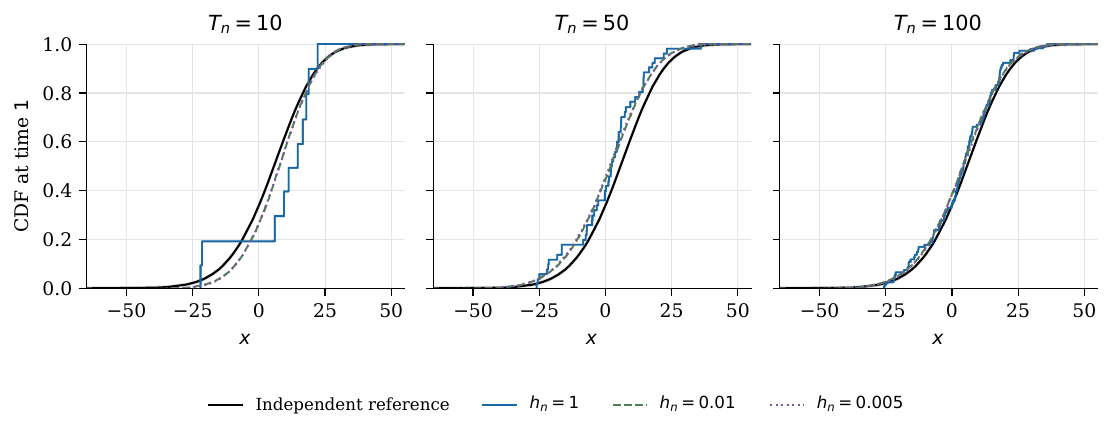}
\caption{CDFs at time 1 and enlarged differences. Columns correspond to the observation horizons. Top: black, independent reference; blue solid, $h_n=1$; green dashed, $h_n=0.01$; purple dotted, $h_n=0.005$. Middle: the green and purple CDFs minus the reference, with its 99\% uniform error band centered at zero (clipped at CDF endpoints). Bottom: $F_{0.01}-F_{0.005}$, the direct difference of the two conditional empirical CDFs. Vertical scales differ between rows but are common across columns within each row. The reference uses 200,000 draws; each conditional CDF uses 20,000 draws.}
\label{fig:marginal-cdf}
\end{figure}

For an additional check, Table~\ref{tab:conditional-moments} gives exact
conditional moments computed from the observed increments, rather than
estimated from the 20,000 resamples. Write $d_j=\Delta_jX$,
$\bar d=n^{-1}\sum_jd_j$, and $s_d^2=n^{-1}\sum_j(d_j-\bar d)^2$.
For the uniform-permutation value at time 1,
\[
 \E[\wh{X}_1^n(I)\mid\mathcal G_n]=u+k\bar d,\qquad
 \Var(\wh{X}_1^n(I)\mid\mathcal G_n)
       =\frac{k(n-k)}{n-1}s_d^2.
\]
These identities follow from sampling $k$ distinct elements of a finite
population. Since $kh_n=1$, the conditional mean is
$u+(X_{T_n}-u)/T_n$, independent of the mesh for a fixed observation
window. Finer sampling alone cannot remove the drift estimation error
of a fixed observed path. For this realization the conditional mean
at $T_n=50$ is farther from 5 than at $T_n=10$; no favorable path was
selected to impose monotonic improvement.

\begin{table}[tbp]
\centering
\caption{Exact conditional moments of the permutation value at time 1.
The unconditional benchmark is $\E X_1=5$ and
$\operatorname{sd}(X_1)=\sqrt{190}\simeq13.784$.}
\label{tab:conditional-moments}
\begin{tabular}{rrrrr}
\hline
$T_n$ & $h_n$ & $n$ & Conditional mean & Conditional SD\\
\hline
10 & 1 & 10 & 7.136 & 15.171 \\
10 & 0.01 & 1000 & 7.136 & 12.351 \\
10 & 0.005 & 2000 & 7.136 & 12.260 \\
50 & 1 & 50 & 0.905 & 13.602 \\
50 & 0.01 & 5000 & 0.905 & 14.366 \\
50 & 0.005 & 10000 & 0.905 & 14.339 \\
100 & 1 & 100 & 3.273 & 13.559 \\
100 & 0.01 & 10000 & 3.273 & 14.087 \\
100 & 0.005 & 20000 & 3.273 & 14.076 \\
\hline
\end{tabular}
\end{table}

The regeneration script, sampled baseline, marginal samples, and numerical
checks accompany the figure files. Checks include preservation of every
aggregated increment sum, the common terminal value of complete
permutations, agreement of the independent reference with the known first
two moments, and agreement of resample means with the exact conditional
moments within Monte Carlo error. These illustrations do not constitute
a simulation study of the stopped-functional M-estimator or establish
its asymptotic normality.

\section{Main results for M-estimation}\label{sec:main}
{\subsection{Consistency results}\label{sec:consist}
Throughout this section, permutations are sampled independently and uniformly
with replacement, independently of the observations, as specified in Section~\ref{sec:quasi-proc}.
Write $a_n$ for the number of draws and $\mathcal H=\{h_\theta:\theta\in\Theta\}$.
The following sufficient conditions explicitly include the approximation of
infinite-horizon functionals by finite-horizon ones.

\begin{thm}\label{thm:consist}
Assume (HFLT), $a_n\to\infty$, and compactness of $\Theta\subset\R^d$.
Suppose that $h_\theta$ is measurable and that, for constants $M,L<\infty$,
\begin{equation}\label{c11b}
 |h_\theta(x)|\le M,\qquad
 |h_\theta(x)-h_{\theta'}(x)|\le L|\theta-\theta'|.
\end{equation}
For an unbounded sequence of deterministic horizons $T$, suppose there are
measurable functionals $h_{\theta,T}$ depending only on $x|_{[0,T]}$, satisfying
the same bounds, continuous at $X|_{[0,T]}$ almost surely in the $J_1$ topology
for every fixed $\theta$, and such that
\[
 \eta_T:=\sup_{x,\theta}|h_\theta(x)-h_{\theta,T}(x)|\longrightarrow0.
\]
For a finite-horizon contrast it suffices to take $h_{\theta,T}=h_\theta$
for all larger $T$. Assume also that, for every $\varepsilon>0$,
\begin{equation}\label{identify}
 \inf_{\theta\in\Theta:\,|\theta-\theta_0|\ge\varepsilon}Ph_\theta
 >Ph_{\theta_0}.
\end{equation}
Then every measurable minimizer in \eqref{M-est} satisfies
$\widehat\theta_n\toP\theta_0$ under the joint law of data and permutations.
\end{thm}

\begin{proof}
Keep the conditional permutation average $\widehat P_n$ distinct from $P$.
The decomposition is
\begin{equation}\label{I1-I2}
 \sup_\theta|(P_n-P)h_\theta|
 \le \underbrace{\sup_\theta|(P_n-\widehat P_n)h_\theta|}_{I_n^{(1)}}
 +\underbrace{\sup_\theta|(\widehat P_n-P)h_\theta|}_{I_n^{(2)}}.
\end{equation}
Conditionally on $\mathcal G_n$, the $a_n$ summands defining $P_n$ are i.i.d.
and bounded by $M$. Thus, for each fixed $\theta$,
\[
 \E\big[|(P_n-\widehat P_n)h_\theta|^2\mid\mathcal G_n\big]
 \le M^2/a_n.
\]
A finite $\delta$-net $\{\theta_j\}_{j=1}^N$ of $\Theta$ gives
\[
 I_n^{(1)}\le\max_{j\le N}|(P_n-\widehat P_n)h_{\theta_j}|+2L\delta.
\]
Chebyshev's inequality and a union bound, first with $n\to\infty$ and then
$\delta\downarrow0$, prove $I_n^{(1)}\toP0$. This argument applies to the
triangular array and does not invoke a fixed-sequence strong law.

We next control the randomness of $\widehat P_n$. Fix $T$ and $\theta$, and set
$m_n=\lfloor T/h_n\rfloor$ and $q_n=\lfloor n/m_n\rfloor$, for $n$ large
enough that $1\le m_n\le n$. The quantity $\widehat P_n h_{\theta,T}$ is the
average of a kernel over all ordered $m_n$-tuples of distinct observed
increments. Equivalently it is a U-statistic with the symmetrized kernel
$K_{n,\theta,T}$, bounded by $M$.

Here is a direct variance bound that remains valid when $m_n$ increases.
For each permutation $\pi$ partition its first $q_nm_n$ indices into $q_n$
disjoint blocks of length $m_n$, and let $B_\pi$ be the average of
$K_{n,\theta,T}$ over these blocks. Averaging $B_\pi$ over all $\pi$ gives
$\widehat P_nh_{\theta,T}$. For fixed $\pi$ the blocks are independent under
the data law, so $\Var(B_\pi)\le M^2/q_n$. Jensen's inequality therefore yields
\begin{equation}\label{eq:conditional-variance}
 \Var(\widehat P_nh_{\theta,T})\le M^2/q_n\longrightarrow0,
 \qquad m_n/n\sim T/T_n\longrightarrow0.
\end{equation}
Moreover, Theorem~\ref{thm:wc}, almost-sure continuity of $h_{\theta,T}$
under the limit law, and boundedness imply
\[
 \E[\widehat P_nh_{\theta,T}]
 =\E[h_{\theta,T}(\wh{X}^n(I_{n,1}))]\longrightarrow Ph_{\theta,T}.
\]
Consequently $(\widehat P_n-P)h_{\theta,T}\toP0$ for each fixed $\theta$.
The same finite-net argument gives uniform convergence in $\theta$.
Finally,
\[
 I_n^{(2)}\le\sup_\theta|(\widehat P_n-P)h_{\theta,T}|+2\eta_T.
\]
Let $n\to\infty$ first and then $T\to\infty$. This proves $I_n^{(2)}\toP0$.
No conditional expectation has been replaced by an unconditional one.
If $|\widehat\theta_n-\theta_0|\ge\varepsilon$, the positive gap in
\eqref{identify} is at most $2\sup_\theta|(P_n-P)h_\theta|$, by minimality.
The probability of this event tends to zero, proving consistency.
\end{proof}

\subsection{Asymptotic normality}\label{sec:asym-normal}
Weak convergence of a single quasi-process does not supply a central limit
theorem for the data-dependent average $\widehat P_n$.
The following result separates this statistical requirement from Monte Carlo
error. Its score limit is an additional assumption; it is not a consequence
of the pathwise approximation condition used in the earlier formulation.

\begin{thm}\label{thm:asymp.normal}
Suppose $\widehat\theta_n\toP\theta_0$ and $\theta_0$ is an interior point of
$\Theta$. On an open ball $B$ about $\theta_0$, suppose $h_\theta$ is twice
continuously differentiable in $\theta$. Put
$s_\theta=\nabla_\theta h_\theta$ and $H_\theta=\nabla_\theta^2h_\theta$.
Assume, uniformly in $x$ and $\theta,\theta'\in B$,
\[
 |s_\theta(x)|+\|H_\theta(x)\|\le C,\qquad
 \|H_\theta(x)-H_{\theta'}(x)\|\le C|\theta-\theta'|.
\]
Suppose $V=PH_{\theta_0}$ is positive definite, $Ps_{\theta_0}=0$, and
\[
 \sup_{\theta\in B}\|\widehat P_nH_\theta-PH_\theta\|\toP0.
\]
For a deterministic sequence $b_n\to\infty$, assume the data-score limit
\begin{equation}\label{rate}
 b_n\widehat P_ns_{\theta_0}\leadsto N(0,\Omega),
 \qquad b_n^2/a_n\longrightarrow0.
\end{equation}
Then
\[
 b_n(\widehat\theta_n-\theta_0)
 \leadsto N(0,V^{-1}\Omega V^{-\top}).
\]
The Hessian law of large numbers follows, for example, by applying the
finite-horizon approximation argument of Theorem~\ref{thm:consist}
componentwise to $H_\theta$ on the closure of a smaller ball.
\end{thm}

\begin{proof}
The conditional variance bound for the score at $\theta_0$ gives
\[
 \E\big[|b_n(P_n-\widehat P_n)s_{\theta_0}|^2\big]
 \le C^2b_n^2/a_n\longrightarrow0.
\]
Thus $b_nP_ns_{\theta_0}\leadsto N(0,\Omega)$.
Boundedness and the Lipschitz property of $H_\theta$, a finite net, and the
conditional variance bound likewise give
$\sup_{\theta\in B'}\|(P_n-\widehat P_n)H_\theta\|\toP0$
on any closed ball $B'\subset B$.
With probability tending to one the minimizer is in the interior of $B'$,
and hence $P_ns_{\widehat\theta_n}=0$. The integral Taylor formula gives
\[
 0=P_ns_{\theta_0}+\overline H_n(\widehat\theta_n-\theta_0),\qquad
 \overline H_n=\int_0^1P_nH_{\theta_0+t(\widehat\theta_n-\theta_0)}\,dt.
\]
The Hessian assumptions and consistency imply $\overline H_n\toP V$.
Therefore
\[
 b_n(\widehat\theta_n-\theta_0)
 =-V^{-1}b_n\widehat P_ns_{\theta_0}+o_p(1),
\]
and Slutsky's theorem proves the assertion.
\end{proof}

\begin{remark}
The covariance $\Omega$ is determined by the observed increments and the
functional. It is not generally $P[s_{\theta_0}s_{\theta_0}^{\top}]$, the
covariance for independent full paths. Increasing $a_n$ reduces Monte Carlo
error but does not create additional independent observations or determine
the statistical rate $b_n$. The previous claims of an arbitrary rate
$\sqrt{r_n}$ and an optimal rate $n^{\beta/(2d)}$ do not follow and are
withdrawn. A model-specific proof of the first condition in \eqref{rate}
is needed before this theorem can be used for inference.
\end{remark}
}

\section{Examples}\label{sec:example}
{\subsection{Default-related discounted losses}\label{sec:ex1}
Let $X$ represent a surplus or asset-value process, with $X_0=u>\xi$.
For a path $x$ on $[0,T]$, where $0<T\le\infty$, set
\[
 \tau_T^x=\inf\{t\in[0,T]:x_t<\xi\}\wedge T,
 \qquad \inf\varnothing=\infty.
\]
The horizon $T$ of the functional is distinct from the observation horizon
$T_n=nh_n$. For $r>0$ define
\begin{equation}\label{dl}
 h_\theta(x)=\int_0^{\tau_T^x}e^{-rt}U_\theta(t,x_t)\,dt,
\end{equation}
and its expectation
\begin{equation}\label{edl}
 Ph_\theta=\E\left[\int_0^{\tau_T^X}e^{-rt}U_\theta(t,X_t)\,dt\right].
\end{equation}
Here $U_\theta$ is measurable in time and the path variable and continuously
differentiable in $\theta$ as many times as specified below.
The discount factor gives a uniform truncation bound $Ce^{-rS}/r$ after
time $S$ whenever $|U_\theta|\le C$.

\begin{defn}\label{def:regset}
For $T<\infty$, define $\DD_T^{\mathrm{reg}}$ to consist of paths $y$ with
$y_0>\xi$ such that
\[
 \inf_{0\le s\le a}(y_s-\xi)>0\quad\hbox{for every }a<\tau_T^y,
\]
and, if $\tau_T^y<T$, for every $\varepsilon>0$ there exists
$s\in[\tau_T^y,\min\{T,\tau_T^y+\varepsilon\})$ with $y_s<\xi$.
The latter condition states immediate strict crossing. The former is
stronger than pointwise strict positivity: a positive c\`adl\`ag function
may have left limit zero immediately before an upward jump.
\end{defn}

\begin{lemma}\label{lem:tau:reg}
For finite $T$, the map $x\mapsto\tau_T^x$ is $J_1$-continuous at every
$y\in\DD_T^{\mathrm{reg}}$, including paths with $\tau_T^y=T$.
\end{lemma}
\begin{proof}
Suppose $x_j\to y$ in $J_1$ and choose increasing continuous bijections
$\lambda_j:[0,T]\to[0,T]$ with
$\|\lambda_j-I\|_T+\|x_j\circ\lambda_j-y\|_T\to0$.
Write $\tau=\tau_T^y$. For $0<\varepsilon<\tau$, positivity on
$[0,\tau-\varepsilon/2]$ has a uniform margin. For large $j$, every
$s\le\tau-\varepsilon$ has $\lambda_j^{-1}(s)\le\tau-\varepsilon/2$,
so $x_j(s)>\xi$. This proves $\liminf_j\tau_T^{x_j}\ge\tau$.
If $\tau<T$, choose a strict crossing point
$s<\min\{T,\tau+\varepsilon/2\}$. For large $j$,
$x_j(\lambda_j(s))<\xi$ and $\lambda_j(s)<\tau+\varepsilon$.
Hence $\limsup_j\tau_T^{x_j}\le\tau$. If $\tau=T$ this upper bound holds
by truncation. Both bounds prove continuity.
\end{proof}

\begin{lemma}\label{lem:h}
Let $N\ge2$. Suppose, for $k=0,\ldots,N$, that
\begin{align}
 \sup_{t,z,\theta}|\nabla_\theta^kU_\theta(t,z)|&\le C,\label{U-bound}\\
 \sup_{t,\theta}|\nabla_\theta^kU_\theta(t,z)-\nabla_\theta^kU_\theta(t,z')|
 &\le L|z-z'|.\label{U-lipschitz}
\end{align}
For finite $T$, $h_\theta$ and its parameter derivatives through order $N-1$
are bounded and continuous in $x$ at each $y\in\DD_T^{\mathrm{reg}}$.
\end{lemma}
\begin{proof}
If $x_j\to y$ in $J_1$, then $x_j(t)\to y(t)$ at every continuity point
of $y$, hence for Lebesgue-almost every $t$. By Lemma~\ref{lem:tau:reg},
$\tau_T^{x_j}\to\tau_T^y$, so
$\mathbf1_{\{t<\tau_T^{x_j}\}}\to\mathbf1_{\{t<\tau_T^y\}}$
for almost every $t$. The integrands in \eqref{dl} converge almost
everywhere and are dominated by $Ce^{-rt}$. Dominated convergence proves
the assertion. Differentiation under the integral is justified by the
same bound; $\tau_T^x$ is independent of $\theta$.
This proof does not differentiate the Skorokhod time change and does not
assert a Lipschitz modulus in the $J_1$ metric.
\end{proof}

\begin{cor}\label{cor:h-class}
Under Lemma~\ref{lem:h}, $h\in C_b^{0,k}(\DD_T^{\mathrm{reg}}\times\Theta)$
for $0\le k\le N-1$, with derivatives understood on an open parameter
neighborhood when $\Theta$ is compact.
\end{cor}

\begin{prop}\label{prop:Ph}
Let $X^n\leadsto X$ in $\DD_T$ with finite $T$, and suppose
$\P(X|_{[0,T]}\in\DD_T^{\mathrm{reg}})=1$.
Under \eqref{U-bound}--\eqref{U-lipschitz},
$\E[h_\theta(X^n)]\to Ph_\theta$ for each fixed $\theta$.
For the uniformly randomized estimator, Theorem~\ref{thm:consist} gives
consistency under its compactness and identification assumptions.
For $T=\infty$, the same consistency conclusion holds if regularity has
probability one at an unbounded sequence of finite horizons.
\end{prop}
\begin{proof}
Apply the continuous mapping theorem to Lemma~\ref{lem:h} and use
boundedness to obtain convergence of expectations. The bound on the
first parameter derivative supplies uniform equicontinuity on a compact
parameter set (assuming the bound on a neighborhood containing its convex
hull). For the infinite-horizon case use the uniform tail bound
$Ce^{-rS}/r$ in Theorem~\ref{thm:consist}.
\end{proof}

\begin{remark}
The probability-one regularity condition is an assumption to be verified
for the process, not a consequence of $X_0>\xi$ alone. It holds, for
example, for Brownian motion with drift and nonzero volatility started
above $\xi$: continuity gives a positive minimum on each compact interval
strictly before the first barrier contact, and the strong Markov property
and immediate Brownian oscillation give strict crossing immediately after
contact. No universal assertion for all jump models is made here.
\end{remark}

\begin{lemma}\label{lem:smooth}
Let $\psi_\delta:\R\to[0,1]$ be smooth and nondecreasing, zero on
$(-\infty,0]$ and one on $[\delta,\infty)$. Write
$F(x;t)=\inf_{0\le s\le t}x_s$, and set
\[
 \widetilde h_\delta(x,\theta)=\int_0^T e^{-rt}U_\theta(t,x_t)
       [1-\psi_\delta(\xi-F(x;t))],dt.
\]
Under \eqref{U-bound}--\eqref{U-lipschitz}, for fixed $\delta>0$ this
functional and its parameter derivatives through order $N-1$ are bounded
and $J_1$-continuous on the whole path space. Moreover,
$\widetilde h_\delta(x,\theta)\to h_\theta(x)$ for every path $x$ as
$\delta\downarrow0$, and the expectations converge as well.
\end{lemma}
\begin{proof}
For finite $T$, if $x_j\to x$ in $J_1$, take corresponding time changes
$\lambda_j$ and errors $e_j=\|x_j\circ\lambda_j-x\|_T\to0$.
Surjectivity on initial intervals yields
\[
 |F(x_j;\lambda_j(t))-F(x;t)|\le e_j.
\]
At continuity points of the monotone function $F(x;\cdot)$, the fact that
$\lambda_j^{-1}(t)\to t$ implies $F(x_j;t)\to F(x;t)$.
Both this convergence and $x_j(t)\to x(t)$ hold for almost every $t$.
Dominated convergence now proves continuity of the integral and of its
parameter derivatives. For $T=\infty$, first restrict to finite horizons
at continuity points of the limit path and then use the uniform
exponentially small tail.
As $\delta\downarrow0$, the bracketed factor tends to
$\mathbf1_{\{F(x;t)\ge\xi\}}$. This agrees with
$\mathbf1_{\{t<\tau_T^x\}}$ except possibly at $t=\tau_T^x$, which does
not affect the integral. Dominated convergence proves the remaining
claims. In particular, no fixed-time continuity of the running infimum
at every c\`adl\`ag path is assumed.
\end{proof}

\subsection{Pathwise bounds and their limitations}
\begin{lemma}[A bound retaining the stopping-time error]\label{lem:nabla-h}
For any two jointly defined paths $Y,X$, any $p\ge1$, and $k\le N-1$,
\begin{align*}
 \|\nabla_\theta^kh_\theta(Y)-\nabla_\theta^kh_\theta(X)\|_{L^p}
 &\le L\int_0^T e^{-rt}\|Y_t-X_t\|_{L^p}\,dt\\
 &\quad+\frac{C}{r}\|e^{-r\tau_T^Y}-e^{-r\tau_T^X}\|_{L^p},
\end{align*}
where $e^{-r\infty}=0$. Thus a rate for the first integral alone does not
establish a rate for the stopped functional.
\end{lemma}
\begin{proof}
Split the difference at the smaller stopping time. Bound the common
part by the spatial Lipschitz constant and the remaining interval by
$C\int_{\tau_T^Y\wedge\tau_T^X}^{\tau_T^Y\vee\tau_T^X}e^{-rt}dt$.
Minkowski's inequality gives the displayed estimate.
\end{proof}

\begin{prop}[Exact error for a uniform permutation]\label{lem:nabla-h2:pj}
Assume $v=\Var(X_1-X_0)\in(0,\infty)$, and let $I$ be uniform on
$\Lambda_n$, independently of $X$. At $t_k=kh_n$,
\begin{equation}\label{eq:exact-error}
 \E|\wh{X}_{t_k}^n(I)-X_{t_k}|^2
   =2v h_nk(1-k/n).
\end{equation}
Consequently, if $t_{k_n}\to t>0$ under (HFLT), this quantity tends to
$2vt$, rather than zero. This also applies to finite-variance pure-jump
processes.
\end{prop}
\begin{proof}
Let $c(I)=\#(\{I(1),\ldots,I(k)\}\cap\{1,\ldots,k\})$ and
$R=k-c(I)$. Common increments cancel. Given $I$, the remaining difference
contains $R$ positively signed and $R$ negatively signed independent
increments. Its mean is zero, since their means cancel, and its variance
is $2Rvh_n$. Since $\E c(I)=k^2/n$, averaging gives \eqref{eq:exact-error}.
In particular, for $p\ge1$,
\[
 \|k-c(I)\|_{L^p}\ge\E[k-c(I)]=k(1-k/n).
\]
The earlier $O(\sqrt{k})$ bound is therefore false when $k/n\to0$.
For a compensated Poisson process the same formula gives an explicit
pure-jump counterexample to the claimed vanishing error.
\end{proof}

\begin{remark}
The Blumenthal--Getoor index alone does not determine all small-time
$L^p$ increment rates. For example, if $N_h$ is Poisson with mean
$\lambda h$, then $\|N_h\|_{L^p}$ has order $h^{1/p}$ for fixed
$p\ge1$ as $h\downarrow0$. Even with a Brownian component, this jump
term dominates $h^{1/2}$ for $p>2$. Also, increment assumptions concern
$X_h-X_0$, not $X_h$ when $X_0=u\ne0$.
\end{remark}

\begin{defn}[A well-defined sparse permutation]\label{def:near-id}
Select each index independently with probability $\varepsilon_n$ to form
a subset $S_n$. Uniformly permute the indices within $S_n$ and fix every
index outside $S_n$. If $\#S_n\le1$, use the identity.
This produces a bijection, independently of the data. Selection and
actual displacement are different events.
\end{defn}

\begin{prop}\label{lem:nabla-h2:bm}
Under the finite-variance assumption of Proposition~\ref{lem:nabla-h2:pj},
the scheme in Definition~\ref{def:near-id} satisfies
\[
 \E|\wh{X}_{t_k}^n(I)-X_{t_k}|^2\le2v t_k\varepsilon_n.
\]
This is an $L^2$ bound and does not assert an analogous bound for all $p$.
\end{prop}
\begin{proof}
The cancellation argument still gives conditional variance $2vh_nR$,
where $R$ is the number of original first-$k$ indices lost from the
permuted first-$k$ set. It is at most $\#(S_n\cap\{1,\ldots,k\})$,
whose mean is $k\varepsilon_n$. Average over the selections and permutations.
An uncentered binomial count $M$ has $\|M\|_{L^p}\ge\E M=k\varepsilon_n$;
an $O(\sqrt{k\varepsilon_n})$ bound cannot hold in general.
\end{proof}

\begin{remark}\label{rem:ex-rate:fixed}
The sparse scheme is not the uniform permutation scheme required in
Theorem~\ref{thm:consist}. In fact, approaching the identity can retain the
randomness of one original path instead of estimating an expectation.
For an explicit example take $X=W$, fix $t>0$, and use the bounded
contrast $h_\theta(x)=(\theta-\tanh x_t)^2$ on $\Theta=[-1,1]$.
Its population minimizer is zero. The minimizer of the resampled contrast
is $P_n(\tanh x_t)$. The preceding $L^2$ bound, the Lipschitz property of
$\tanh$, and $W_{\lfloor t/h_n\rfloor h_n}\to W_t$ in $L^2$ imply
\[
 \widehat\theta_n-\tanh W_t\toP0
\]
as $\varepsilon_n\to0$, regardless of the number of draws. Thus consistency
for zero fails. Closeness to the observed path cannot replace a law of
large numbers for the conditional permutation average.
\end{remark}

\begin{remark}\label{cor:ex-normal:final}
The earlier normality corollary for discounted default losses, based on
vanishing uniform-permutation pathwise error and the rate $n^{\beta/(2p)}$,
is withdrawn. Proposition~\ref{prop:Ph} supplies consistency under its
explicit regularity and identification conditions. To obtain asymptotic
normality for that loss, one must additionally prove the data-score limit
and Hessian conditions of Theorem~\ref{thm:asymp.normal}. Neither the
corrected stopping-time bound nor smoothing alone proves those conditions.
\end{remark}

\subsection{A directly verifiable normal limit}
An elementary contrast makes the distinction between observation error
and permutation error explicit. This example is unbounded and is proved
directly, rather than by invoking the bounded-contrast theorems above.
Assume $\E(X_1-u)=m$ and $\Var(X_1-u)=v\in(0,\infty)$. Fix $t>0$ and take
\[
 h_\theta(x)=(x_t-u-\theta t)^2,\qquad\Theta=\R.
\]
The unique population minimizer is $m$. For uniform independent draws,
put $k_n=\lfloor t/h_n\rfloor$ and $t_n^*=k_nh_n$. The exact empirical
minimizer and its conditional expectation are
\[
 \widehat\theta_n=\frac{1}{a_nt}\sum_{b=1}^{a_n}
           (\wh{X}_t^n(I_{n,b})-u),\qquad
 \E[\widehat\theta_n\mid\mathcal G_n]
   =\frac{t_n^*}{t}\frac{X_{T_n}-u}{T_n}.
\]
By conditional independence and the total variance inequality,
\[
 \E|\widehat\theta_n-\E[\widehat\theta_n\mid\mathcal G_n]|^2
 \le\frac{v t_n^*}{a_nt^2}.
\]
Consequently, if $T_n/a_n\to0$ and $\sqrt{T_n}h_n\to0$, the usual
finite-variance L\'evy central limit theorem gives
\[
 \sqrt{T_n}(\widehat\theta_n-m)\leadsto N(0,v).
\]
Indeed, the Monte Carlo term is $o_p(T_n^{-1/2})$ and the rounding bias is
$O(h_n)$. The remaining limit is
\[
 (X_{T_n}-u-mT_n)/\sqrt{T_n}\leadsto N(0,v).
\]
The last limit follows from the i.i.d. central limit theorem for unit
increments; the remaining increment of length less than one is negligible
in $L^2$ after centering. The rate here is $\sqrt{T_n}$ and comes from
the observation horizon, not from an arbitrary choice of the number of
permutations.

\subsection{Dividends up to ruin}
To represent dividends rather than losses, maximize the expected reward,
or minimize its negative. Take $a>0$ and a smooth strictly increasing
maturity function $g$, and define
\[
 R_\theta(x)=a\int_0^{\tau_T^x}e^{-rt}
       \mathbf1_{\{x_t\ge\theta\}}\mathbf1_{\{t\le g(\theta)\}}\,dt,
 \qquad h_\theta=-R_\theta.
\]
For a smooth function $\varphi_\varepsilon:\R\to[0,1]$ equal to zero on
$(-\infty,-\varepsilon]$ and one on $[\varepsilon,\infty)$, use
\[
 R_\theta^\varepsilon(x)=a\int_0^{\tau_T^x}e^{-rt}
   \varphi_\varepsilon(x_t-\theta)
   \varphi_\varepsilon(g(\theta)-t)\,dt.
\]
This avoids an inverse of $g$ outside its range. For fixed $\varepsilon$
and bounded derivatives of $g$ on the parameter set, the smooth integrand
satisfies the boundedness and spatial Lipschitz assumptions. Constants
may depend on $\varepsilon$ and need not stay bounded as
$\varepsilon\downarrow0$.
If
\[
 \int_0^T e^{-rt}\mathbf1_{\{x_t=\theta\}}\,dt=0,
\]
then $R_\theta^\varepsilon(x)\to R_\theta(x)$ by dominated convergence;
the single maturity endpoint has zero Lebesgue measure. This convergence
is not asserted for every c\`adl\`ag path, since a path may spend positive
time at the threshold.
For each fixed smoothing level, consistency follows under
Proposition~\ref{prop:Ph} and identification of the smoothed target.
Asymptotic normality requires the additional score and Hessian conditions
in Theorem~\ref{thm:asymp.normal}. Passing to the unsmoothed optimizer also
requires a separate control of smoothing bias and identification; it
does not follow merely from pointwise convergence of rewards.
The numerical study below instead retains the original indicators and
directly investigates the unsmoothed estimator; its interpretation and
the distinction from these sufficient conditions are stated in
Remark~\ref{rem:unsmoothed-experiment}.
}

\section{Simulation study for the dividend criterion}\label{sec:dividend-simulation}

We examine the finite-sample behavior of the quasi-process M-estimator
for the unsmoothed dividend criterion in Section~\ref{sec:example}.
We retain the indicator functions defining the original reward and study
numerical behavior relevant to consistency and a possible normal
approximation: concentration around independent numerical benchmarks,
changes in dispersion, and departures from a continuous sampling law.
These are finite-sample investigations, rather than applications of the
asymptotic theorems to this nonsmooth criterion.
This experiment uses the jump-diffusion model
\[
 X_t=20+20t+10W_t-\sum_{j=1}^{N_t}J_j,
\]
where $W$ is a standard Brownian motion, $N$ is a Poisson process with
intensity $5$, and the $J_j$ are exponentially distributed with mean $3$;
these objects are independent. The initial surplus is $20$ rather than
the value $0$ used for the path illustrations in Section~\ref{sec:weak}.
The ruin boundary is $0$. We choose dividend rate $a=1$, discount rate
$r=0.2$, maturity function $g(\theta)=\theta/5$, parameter set
$\Theta=[0,60]$, and evaluation horizon $T=12$. Thus
\[
 R_\theta(x)=\int_0^{\tau_{12}^x}e^{-0.2t}
       \mathbf1_{\{x_t\geq\theta\}}
       \mathbf1_{\{t\leq\theta/5\}}\,dt.
\]
Dividends do not change the underlying surplus path. These numerical
choices specify this experiment.

\begin{remark}[Scope of the unsmoothed experiment]\label{rem:unsmoothed-experiment}
The uniform parameter-Lipschitz condition in Theorem~\ref{thm:consist}
does not hold for this reward on the whole path space. For example,
on the constant path $x_t=20$, the present specification gives
$R_{20}(x)=(1-e^{-0.8})/0.2$, whereas $R_{20+\delta}(x)=0$ for every
$0<\delta\le40$. Thus the sample contrast need not even be continuous
in the parameter. The differentiability assumptions of
Theorem~\ref{thm:asymp.normal} are also not available.
These sufficient conditions failing does not imply that the estimator
is inconsistent or has no normal limit; it means that the preceding
theorems do not settle those questions for this criterion.
We use the original, unsmoothed reward because it is the quantity of
practical interest, and assess its behavior numerically. Establishing
consistency or asymptotic normality for this unsmoothed estimator
requires a separate argument and is not claimed here.
\end{remark}

\subsection{Design and numerical benchmarks}
We use observation horizons $H\in\{25,100,400\}$. The observation
spacing takes the values $h=0.2$, $0.1$, and $0.05$, giving $n=H/h$ increments.
For each data set, we draw $B\in\{100,500,2000\}$ independent uniform
permutations, with replacement between permutations, and construct the
corresponding step paths. Only the ordered prefix needed on $[0,12]$
is generated, which has the same law as the prefix of a full permutation.
We maximize the average reward over
$\Theta_{\rm grid}=\{0,0.1,\ldots,60\}$, choosing the smallest maximizer
in a tie. The discounted integral on each step interval, including a
maturity lying inside that interval, is evaluated exactly. Ruin is
detected on the step path; crossings between observations are not
reconstructed.

Each of the $27$ configurations is repeated on $R=200$ independent
data sets. Within a replication, observation designs share a base
increment sequence, and the values of $B$ use nested prefixes of the
same permutation sequence. Hence comparisons across designs are paired,
whereas replications are independent. The reported standard deviation is
the sample standard deviation of the $R$ estimates, with denominator
$R-1$, rather than the standard error of their mean.

To obtain numerical benchmarks independently of the estimation runs,
we simulate $100{,}000$ paths at spacing $h_*=0.0125$, also aggregating
them to the coarser spacings. We maximize the resulting mean reward on
the same parameter grid. The benchmark maximizers at spacings
$0.2$, $0.1$, $0.05$, $0.025$, and $0.0125$ are respectively
$20.0$, $19.7$, $19.2$, $18.9$, and $18.8$.
We denote these estimates by $\widetilde\theta_h$.
They are Monte Carlo, time-discretized benchmarks, not known exact
continuous-time parameters. Resampling $20$ independent groups of
reference paths $1000$ times gives a bootstrap standard deviation of
$0.0324$ for $\widetilde\theta_{h_*}$; this does not account for time
or parameter discretization error.

For estimates $\widehat\theta_1,\ldots,\widehat\theta_R$, we report
\[
 \overline\theta=R^{-1}\sum_{j=1}^R\widehat\theta_j,
 \qquad
 D_h=\overline\theta-\widetilde\theta_h,
 \qquad D_*=\overline\theta-\widetilde\theta_{h_*}.
\]
These signed differences estimate bias relative to the specified
numerical benchmarks. They must not be interpreted as bias relative to
an exactly known continuous-time optimizer. In particular, $D_*$ also
reflects the effect of changing the observation spacing.

\subsection{Empirical concentration and dispersion}
Table~\ref{tab:dividend-estimation} reports the finest observation
spacing. At $H=400$, increasing $B$ from $100$ to $2000$ reduces the
standard deviation from $0.7722$ to $0.4844$. At $B=2000$, increasing
$H$ from $25$ to $400$ reduces it from $2.0360$ to $0.4844$.
The corresponding mean at $H=400$ is $19.1530$, close to the same-spacing
benchmark $19.2$, but $0.3530$ above the finer-spacing benchmark $18.8$.
Increasing the number of permutations alone does not remove time
discretization effects. The reduction in dispersion is compatible with
concentration as more observations are used, but the benchmark differences
need not decrease monotonically. In particular, increasing $H$ at fixed
$h$ does not test convergence to the continuous-time optimizer. Such a
consistency statement would also require $h\to0$, adequate permutation
sampling, and control of the parameter-grid error. The fixed grid used
here is part of the numerical approximation, not a vanishing-mesh
asymptotic scheme.

\begin{table}[!htbp]
\caption{Dividend-criterion estimates at $h=0.05$, based on $200$
independent data replications. Here $D_h$ uses the benchmark $19.2$
and $D_*$ uses $18.8$. Both benchmarks are numerical approximations.}
\label{tab:dividend-estimation}
\centering
\begin{tabular}{rr rrrr}
\hline
$H$ & $B$ & Mean & SD & $D_h$ & $D_*$\\
\hline
25 & 100 & 19.1580 & 2.5455 & -0.0420 & 0.3580\\
25 & 500 & 19.1720 & 2.2281 & -0.0280 & 0.3720\\
25 & 2000 & 19.1390 & 2.0360 & -0.0610 & 0.3390\\
100 & 100 & 18.9650 & 0.9708 & -0.2350 & 0.1650\\
100 & 500 & 19.0335 & 0.8359 & -0.1665 & 0.2335\\
100 & 2000 & 19.0250 & 0.7648 & -0.1750 & 0.2250\\
400 & 100 & 19.0135 & 0.7722 & -0.1865 & 0.2135\\
400 & 500 & 19.1460 & 0.5911 & -0.0540 & 0.3460\\
400 & 2000 & 19.1530 & 0.4844 & -0.0470 & 0.3530\\
\hline
\end{tabular}
\end{table}

\begin{figure}[!htbp]
\centering
\includegraphics[width=\linewidth]{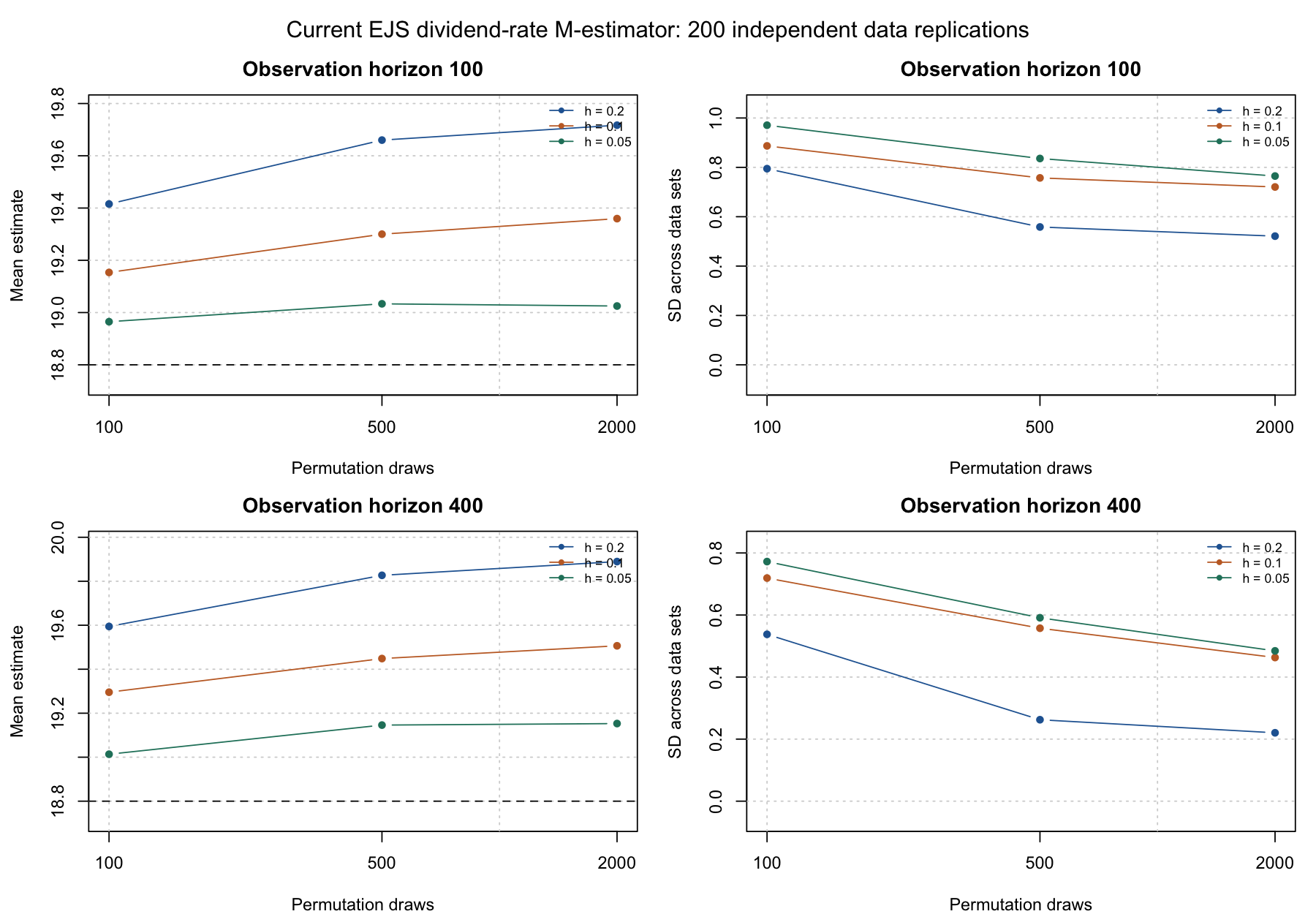}
\caption{Mean and standard deviation across $200$ replications for
the dividend M-estimator. Dashed horizontal lines in the mean panels
show the independent fine-spacing benchmark $18.8$, not an exact
continuous-time target.}\label{fig:dividend-mean-sd}
\end{figure}

Figure~\ref{fig:dividend-mean-sd} compares all three observation
spacings for $H=100$ and $400$. The smaller spread at a coarse spacing
should not be confused with improved accuracy for the continuous-time
target: the numerical benchmark itself changes substantially with $h$.
\subsection{Implications for a normal approximation}
At $H=400$, $h=0.05$, and $B=2000$, $6\%$ of estimates equal the initial
surplus $20$. The step-path approximation retains the initial level
over its first interval, and the unsmoothed threshold can produce
concentration at that value. The fixed parameter grid also makes the
reported estimates discrete. These features caution against assuming a
continuous normal approximation at the sample sizes considered.
Decreasing standard deviations alone do not establish a normal limit,
and the reported means and standard deviations do not determine the
shape of the sampling distribution. No asymptotic rate or limiting
variance for this unsmoothed estimator is inferred from these summaries.
The experiment therefore documents both concentration and finite-sample
features that may obstruct a normal approximation, leaving asymptotic
normality unresolved. In particular, it does not verify the score and
Hessian assumptions of Theorem~\ref{thm:asymp.normal}.

\subsection{Reproducibility and numerical checks}
All simulations, estimations, summaries, and figures were produced in
R~4.4.1 without additional packages. We used the L'Ecuyer--CMRG generator,
with seeds $20260915$ for the estimation experiment, $20260916$ for the
independent benchmarks, and $20260917$ for their group bootstrap.
Separate streams were assigned to independent replications and
substreams to permutation draws. Four pilot replications agreed exactly
between serial and parallel execution and with the corresponding
replications of the full run.

Implementation checks included a constant-path analytic integral,
agreement with direct reward evaluation, equality at the dividend
threshold, and permanent cessation of payments after ruin. For the
first ten replications at $H=400$, $h=0.05$, $B=2000$, refining the
parameter grid from $0.1$ to $0.05$ changed the maximizer by at most
$0.05$. This local diagnostic is not a uniform discretization-error
bound. The accompanying R script and saved results record all
configurations, replication-level estimates, and random-stream seeds.

\section{Concluding remarks}\label{sec:conclude}

{The quasi-process construction supplies an unconditional approximation of the L\'evy path law. For independent uniform permutations, a finite-horizon U-statistic variance bound and a uniform tail approximation yield consistency of the minimum contrast estimator under identification. Conditional independence of the Monte Carlo draws must be distinguished from dependence through the common observed data.

A normal limit additionally requires a central limit theorem for the data-dependent conditional score. We have stated this requirement explicitly and verified a normal limit for a quadratic endpoint contrast. The former general pathwise error bounds and the resulting default-loss normality corollary are invalid; no universal or optimal rate for the stopped loss and dividend examples is claimed here. Establishing a score expansion, its covariance, and any smoothing-bias rate for those examples remains further work.

The numerical experiments use a common baseline, exact grid sampling, conditional moment checks, and an independent marginal reference. The dividend-criterion study in Section~\ref{sec:dividend-simulation} additionally reports finite-sample estimator means and standard deviations, distinguishing numerical benchmark differences from exact-target bias. For the original unsmoothed criterion, the observed reduction in dispersion is compatible with concentration, while discretization effects and concentration at the initial surplus limit a normal approximation. The uniform Lipschitz and differentiability assumptions are not verified for this criterion, so the experiments explore behavior beyond the direct scope of our theorems. Neither consistency nor asymptotic normality of the unsmoothed estimator is established by these finite-sample results. Extensions to multivariate processes and other path-dependent applications require their own identification and limit arguments.
}

\begin{flushleft}
{\bf Acknowledgments}\ \ 
The first author was partially supported by JSPS KAKENHI Grant Numbers JP21K03358 and JST CREST JPMJCR14D7, Japan. 
The second author was supported by JSPS KAKENHI Grant Number JP16K00036.
\end{flushleft}


\begin{thebibliography}{99}
\bibitem{af03} Adams, R. A. and Fournier, J. J. F. (2003). {\it Sobolev spaces}. 2nd ed. Elsevier/Academic Press, Amsterdam. 

\bibitem{Applebaum09} Applebaum, D. (2009). \emph{L\'evy Processes and Stochastic Calculus} (2nd ed.). Cambridge University Press.

\bibitem{b99} Billingsley, P. (1999). {\it Convergence of probability measures}. 2nd ed. John Wiley \& Sons, New York. 

\bibitem{ContTankov04}
Cont, R. and Tankov, P. (2004). \emph{Financial Modelling with Jump Processes}. Chapman \& Hall/CRC, CRC Financial Mathematics Series.

\bibitem{f11} Feng, R. (2011). An operator-based approach to the analysis of ruin-related quantities in jump diffusion risk models. {\it Insurance: Mathematics and Economics}, {\bf 48} (2), 304--313. 

\bibitem{fs13} Feng, R. and Shimizu, Y. (2013). On a generalization from ruin to default in a L\'evy insurance risk model. {\it Methodol. Comput. Appl. Probab.}, {\bf 15}, (4), 773--802. 

\bibitem{cg11} Comte, F. and Genon-Catalot, V. (2011). Estimation for L\'evy processes from high frequency data within a long term interval. {\it The Annals of Statistics}, {\bf 39}, (2), 803--837. 

\bibitem{ev10} Evans, L. C. (2010). {\it Partial Differential Equations}, 2nd. ed. American Mathematical Society, Providence, RI. 

\bibitem{gs97} Gerber, H. U. and Shiu, E. S. W. (1997). From ruin theory to pricing reset guarantees and perpetual put options, {\it Insurance: Math. and Econom.}, {\bf 24}, 3--14. 

\bibitem{j07} Jacod, J. (2007). Asymptotic properties of power variations of L\'evy processes. {\it ESAIM Probab. Statist.} {\bf 11}, 173--196.

\bibitem{kk20} Kato, K. and Kurisu, D. (2020). Bootstrap confidence bands for spectral estimation of L \'evy densities under high-frequency observations. {\it Stoch. Proc. Appl.}, {\bf 130}, (3), 1159--1205. 

\bibitem{k66} Kuratowski, K. (1966). {\it Topology I}. Academic Press, New York.

\bibitem{m92} Mammen, E. (1992). {\it When Does Bootstrap Work?: Asymptotic Results and Simulations}, Springer-Verlag, New York. 

\bibitem{Sato99}
Sato, K. (1999). \emph{L\'evy Processes and Infinitely Divisible Distributions}. Cambridge University Press, Cambridge Studies in Advanced Mathematics, Vol. 68.

\bibitem{s09} Shimizu, Y. (2009). Functional estimation for L\'evy measures of semimartingales with Poissonian jumps. {\it J. Multivariate Anal.}, {\bf 100}, (6), 1073--1092. 

\bibitem{ss21} Shimizu, Y. and Shiraishi, H. (2023). Semiparametric Estimation of Optimal Dividend Barrier for Spectrally Negative L\'evy Process, {\it Research Papers in Statistical Inference for Time Series and Related Models}, Springer Nature Link, 497--517. 

\bibitem{v98} van der Vaart, A. W. (1998). {\it Asymptotic Statistics}. Cambridge University Press, Cambridge.

\bibitem{vw96} van der Vaart, A. W. and Wellner, J. A. (1996). {\it Weak Convergence and Empirical
Processes: With Applications to Statistics}. Springer, New York. 


\bibitem{massart90} Massart, P. (1990). The tight constant in the Dvoretzky--Kiefer--Wolfowitz inequality. {\it The Annals of Probability}, {\bf 18}(3), 1269--1283.

\end{thebibliography}
\end{document}